\documentclass{article}

\usepackage{amsmath}
\usepackage{amssymb}
\usepackage{amsthm}

\usepackage{Baskervaldx}
\usepackage[baskervaldx]{newtxmath}

\usepackage{enumitem}
\usepackage{graphicx}
\usepackage{subcaption}

\usepackage{mathrsfs}
\usepackage{scrextend}
\usepackage[scr=boondoxo,scrscaled=1.05]{mathalfa}
\usepackage{pgf}

\usepackage{adjustbox}

\KOMAoption{fontsize}{14pt}

\usepackage{sectsty}
\sectionfont{\normalsize}
\subsectionfont{\normalsize}
\usepackage{titlesec}

\usepackage{geometry}
\usepackage{xcolor}
\usepackage[colorlinks = true, citecolor = red]{hyperref}
\usepackage{pagecolor}

\newcommand{\reals}{\mathbb R}
\newcommand{\h}{H}
\newcommand{\scl}{\mathscr l}
\newcommand{\partition}{\mathscr P}
\DeclareMathOperator{\mesh}{mesh}

\newtheorem{theorem}{Theorem}[section]

\theoremstyle{definition} 
\newtheorem{property}{Property}

\title{How much work can you get\\
by removing weights from a piston?}

\author{%
\begin{tabular}{c}
\small Joshua Samani\\
\small Department of Physics and Astronomy\\
\small University of California, Los Angeles\\
\small \texttt{jsamani@physics.ucla.edu}
\end{tabular}
}

\date{\small\today}

\begin{document}

\maketitle

\begin{abstract}
In thermodynamics, reversible adiabatic expansion can be understood as a limit of stepwise irreversible processes.  We make this idea concrete by studying an ideal gas in an insulating cylinder with a frictionless piston supporting a load divided into $N$ blocks.  If blocks are removed one at a time, the resulting expansion is a stepwise, irreversible process, but we prove that for a fixed total load, the work done by the gas approaches the reversible limit as the largest block mass tends to zero.  On the way to this limit, an interesting work optimization question arises at finite $N$: how does the work done by the gas depend on the order and sizes of the removed blocks?  Guided by numerical experiments accessible to advanced undergraduates, we motivate and then prove general answers to these questions.  Our main finite-$N$ result proves that for fixed $N$, the optimal stepwise expansion corresponds to a geometric progression of equilibrium pressures, settling a conjecture previously made by Andresen, Berry, Nitzan, and Salamon.
\end{abstract}

\section{Introduction}

The reversible, adiabatic expansion of a gas is often introduced as an ideal limiting process in which the external pressure on the gas is changed infinitesimally slowly, so the gas remains arbitrarily close to equilibrium throughout.  Our primary goal in this paper is to make this limiting process computationally explicit with a simple model system.  This goal leads naturally to a number of interesting optimization questions about irreversible stepwise expansions with finite numbers of steps.  We explore answers to these questions with numerical experiments that can be used as exercises for advanced undergraduate students, allowing them to discover the main patterns for themselves.  We then state and prove mathematical theorems that formalize these patterns.

To make the approach to reversibility concrete, we follow Van Ness \cite{VanNess1983} and consider the stepwise adiabatic expansion of a gas trapped in a cylinder with a piston (Fig. \ref{step}).  
\begin{figure}[h!]
	\centering\includegraphics[width = 0.95\textwidth]{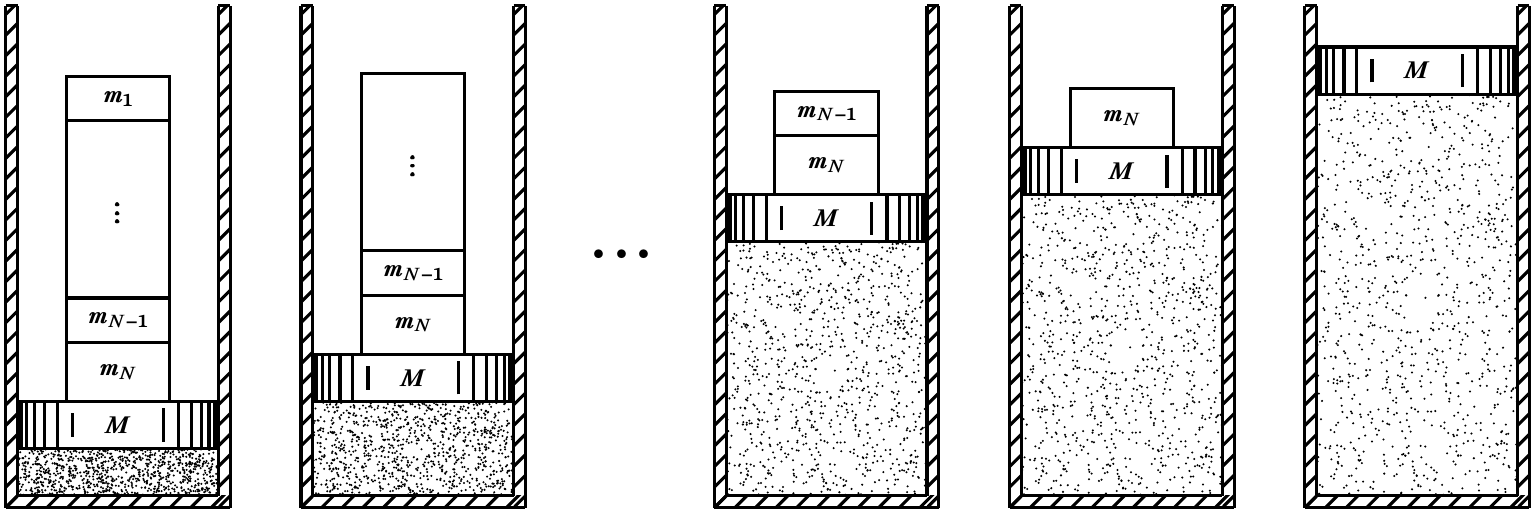}
	\caption{A concrete model for approaching reversible adiabatic expansion: a thermally insulated gas expands by lifting a frictionless piston as \(N\) stacked blocks are removed sequentially.}
	\label{step}
\end{figure}
Both the cylinder and piston are thermally insulating, and the piston slides without friction.  The piston has mass $M$, and on top of it sits a block of total mass $m$ partitioned into a stack of smaller blocks $m_1, \dots, m_N$.  The top block $m_1$ is removed quickly.  This causes the downward force on the piston, and hence the external pressure on the gas, to drop abruptly, so the gas expands.  After some transient oscillations, the piston relaxes to a new equilibrium height through its coupling to the gas, as discussed in Sec.~\ref{sec:equilibration}.  The next block $m_2$ is removed quickly, and again the system is given time to equilibrate at a new height, and so on.  After all blocks are removed, only the piston remains, and the gas attains a final equilibrium.

In this model, the approach to reversibility arises by imagining a very fine partition, one with many small blocks.\footnote{In his original treatment, Van Ness considers a pile of sand, but we consider a stack of blocks for computational simplicity.} Van Ness argues qualitatively that the corresponding expansion is close to reversible, with the reversible expansion yielding an upper bound on the work done by the gas that can be approached in the limit where the largest block mass tends to zero.  Van Ness's discussion is conceptually illuminating, but he does not mathematically show, starting from a finite-step model, that the work approaches the reversible value.  We fill this gap by proving that the reversible work is a strict upper bound for every finite stepwise expansion, and that the stepwise work approaches this upper bound whenever the mesh of the partition, defined as the largest block mass, tends to zero.  

This convergence result naturally raises a finite-\(N\) optimization question: for fixed \(N\) and fixed total load on the piston, which partition yields the maximum work?  An answer to this question is stated and partially proven by Andresen, Berry, Nitzan, and Salamon \cite{Andresen1977} for an ideal gas.  They show that, for fixed endpoint pressures \(P_0\) and \(P_N\), the work viewed as a function of the equilibrium pressures \(P_0,P_1,\dots,P_N\) has a unique critical point in which the pressures form a geometric progression:
\begin{align}\label{geop}
	P_n = \alpha P_{n-1}, \qquad \alpha = (P_N/P_0)^{1/N}.
\end{align}
Although they suggest this partition is a global maximizer of the work among all partitions of length $N$, they explicitly check only the first-order optimality condition, and global optimality remains a conjecture.  We provide a rigorous, elementary proof of global optimality in Theorem \ref{abns} by transforming the problem into a statement about convex functions.   Relatedly, Anacleto, Ferreira, and Soares \cite{Anacleto2009limit} showed that for both geometric and arithmetic progressions of pressures, the entropy production of the stepwise expansion goes to zero as $N\to\infty$.  Anacleto and Ferreira \cite{Anacleto2009finite} extended these results by showing that for finite $N$, the geometric progression of pressures minimizes entropy production. 

In this paper, we motivate and prove a collection of results that complement and extend these observations.  In Sec.~\ref{prep}, we show that maximizing the stepwise work is equivalent to minimizing a dimensionless height-ratio function \(H\) of the block masses $m_1,\dots,m_N$.  In Sec.~\ref{explore}, we use numerical experiments to explore how \(H\) depends on the order and sizes of the blocks.  These experiments lead to conjectures that are proven in Appendix~\ref{proofs}: for a fixed multiset of block masses, work is maximized by removing the blocks in non-increasing order; for fixed \(N\) and fixed total mass, the unique optimal partition corresponds to a geometric progression of successive equilibrium pressures; and as the mesh (largest block) of the partition tends to zero, the stepwise process approaches the reversible adiabatic value.  Together, these results give a concrete finite-step route to the reversible adiabatic limit and show how optimal finite-$N$ stepwise expansions fit into that limiting process.

\section{Modeling assumptions and preparatory calculations}\label{prep}

We assume the gas in the cylinder is ideal, so when it is in equilibrium, the following relationship holds between its internal energy $U$, absolute pressure $P$, and volume $V$:
\begin{align}\label{upv}
	U = cPV.
\end{align}
Here $c = C_V/(N_\mathrm{gas}k_B) > 0$ is the dimensionless constant-volume heat capacity per particle and $N_\mathrm{gas}$ is the number of particles in the gas. The whole apparatus is in a vacuum, so atmospheric pressure can be neglected.  The piston and masses above it are under the influence of a constant gravitational field of magnitude $g$, so in equilibrium, the force exerted on the gas by the piston is $M_\mathrm{tot} g$, where $M_\mathrm{tot}$ denotes the total mass of the piston plus any mass sitting on it.  

Taking the upward direction along the axis of the cylinder to be positive, the piston experiences a net vertical force $PA - M_\mathrm{tot}g$, where $A$ is the cross-sectional area of the cylinder.  In equilibrium, these forces sum to zero, so $P = M_\mathrm{tot}g/A$. If $h$ denotes the height of the bottom of the piston above the bottom of the cylinder, then $V=Ah$.  Combining these observations with Eq. \eqref{upv} gives an expression for the internal energy of the gas entirely in terms of $c$, $M_\mathrm{tot}$, $g$, and $h$;
\begin{align}\label{uhrel}
    U &= cM_\mathrm{tot}gh.
\end{align} 
All analyses that follow depend on this relationship.  In our model system, the ideal gas assumption coupled with the condition for mechanical equilibrium implies that internal energy of the gas is proportional to the gravitational potential energy of the load compressing it.  In what follows, we use Eq. \eqref{uhrel} and the first law of thermodynamics to compute an expression for the work done during any expansion of the gas starting and ending in equilibrium.  In the next section, we take a slight detour we address a
subtle but important point: how can the piston come to rest at a new equilibrium height
after a block is removed if the piston is frictionless and the gas is ideal?  

\subsection{Equilibration}\label{sec:equilibration}

In our model, the phrase ``frictionless piston''
refers to the absence of mechanical friction at the piston-cylinder interface.  It does
not mean that macroscopic motion of the piston cannot be damped by its coupling
to the gas.  Prior work, using both numerical simulation and kinetic theory, has shown that even for an ideal gas, elastic collisions between the gas and piston can be sufficient to damp the piston and produce relaxation to macroscopic equilibrium.

Cerino et al. \cite[Sec.~3]{Cerino2016} study $N$ particles in a cylinder with a piston subject to a constant external force. Particles collide elastically with the walls and piston, so total mechanical energy of the piston and particles is conserved.  They find macroscopic relaxation of the piston to an equilibrium height even in the case where the gas particles do not interact directly with one another; the moving piston itself mediates energy exchange among the gas particles and causes the piston's oscillations to damp out.  Only thermal fluctuations remain, but those vanish in the thermodynamic limit.  A related study is carried out by Boozer \cite{Boozer2008} for a one-dimensional ideal gas interacting elastically with a piston subject to a constant external force.  He numerically shows both an increase in entropy and the relaxation of the particle momentum distribution toward the Maxwell distribution, signaling thermalization.  Taken together, these results show that elastic collisions between the gas and piston can be sufficient for equilibration, even in the absence of friction, heat flow to the gas's environment, or particle-particle interactions.  The same basic point appears in the closely related adiabatic piston problem in which a movable insulating piston separates two ideal gases.  Studies of that system show that the piston can come to rest in mechanical equilibrium through its interaction with the gases, even without mechanical friction at the piston-cylinder interface \cite{gruber2002,chernov2002,gislason2010}.

Mungan \cite{Mungan2017} gives a complementary macroscopic treatment of damped oscillations of a massive piston in an adiabatic cylinder, using the moving-piston pressure correction of Bauman and Cockerham \cite{bauman1969}, without adding kinetic friction on the piston or internal gas viscosity. The physical mechanism for equilibration arises solely from elastic collisions between the gas particles and the piston.  When the piston moves toward the gas, particles colliding
with it rebound with greater speed than they would from a stationary piston, so the
dynamic pressure on the piston exceeds the bulk pressure.  When the piston moves away
from the gas, particles rebound with reduced speed, and the dynamic pressure falls below the
bulk pressure.  The resulting force opposes the piston's motion and results in a damped oscillation.  The kinetic energy of the piston and load is
redistributed among the microscopic kinetic degrees of freedom of the gas.  No heat
crosses the thermally insulating walls or piston.  Total energy of the gas-piston-load system is conserved,
but it is redistributed among the gas internal (kinetic) energy, piston kinetic energy, and gravitational
potential energy of the piston and remaining load.

In the calculations below we do not model the transient oscillatory relaxation in detail.
We use only the initial and final equilibrium states of each step.  At the beginning and
end of a step the piston is at rest in mechanical equilibrium.  This is what enabled us to
express the gas internal energy in terms of the piston height in Eq.~\eqref{uhrel}.  The
mechanical work done by the gas between equilibrium states is stored as additional
gravitational potential energy of the piston and the remaining load while the internal energy of the gas decreases.  Relative to the reversible expansion between the same initial and final loads, the work not captured as gravitational potential energy remains as additional internal energy of the gas.

\subsection{Work between equilibrium states}

We now turn to computing the work between equilibrium states.  Suppose the gas undergoes a process starting and ending in equilibrium.  Let $U_0$ and $U_f$ be the internal energies of the gas in its initial and final equilibrium states respectively.  Since no heat flows between the gas and its environment, the first law of thermodynamics $U_f - U_0 = Q - W$ implies that the total work $W$ performed by the gas is equal to and opposite the change in its internal energy, so the work expressed as a fraction of the initial internal energy is
\begin{align}\label{wf}
	\frac{W}{U_0} 
	&= \frac{-(U_f - U_0)}{U_0} 
	= 1 - \frac{U_f}{U_0}.
\end{align}
If we now assume, as is the case in the stepwise adiabatic expansion, that in the initial state $M_\mathrm{tot} = M + m$ while in the final state $M_\mathrm{tot} = M$, then upon substituting Eq. \eqref{uhrel} for the initial and final internal energies of the gas, we find that
\begin{align}\label{workfrac}
	\frac{W}{U_0} 
	&= 1 - \frac{cMgh_f}{c(M+m)gh_0} 
	= 1-\frac{M}{M+m}\frac{h_f}{h_0} .
\end{align}
This implies that for fixed $M$ and $m$, the work fraction $W/U_0$ increases when $h_f/h_0$ decreases -- more work is done by the gas when the final height of the piston is smaller.  This makes sense because the more work the gas performs, the less internal energy it keeps, but since according to Eq. \eqref{uhrel} the final internal energy is $cMgh_f$, less internal energy at the end means a smaller height $h_f$.  

Equation \eqref{workfrac} reduces the problem of analyzing the work done by the gas to that of analyzing the ratio $h_f/h_0$.  In the next subsection we derive an expression for $h_f/h_0$ in a stepwise expansion.  In the subsection after that, we use standard thermodynamic arguments about reversible processes to derive a candidate expression for a lower bound on $h_f/h_0$.

\subsection{Height ratio for the stepwise expansion}

Consider the stepwise expansion described in the introduction.  After $n$ blocks are removed, the piston will come to equilibrium at some height $h_n$, and the total remaining mass compressing the gas will be
\begin{align}\label{remain}
	M_n = m_{n+1} + \cdots + m_N + M.
\end{align}
For the expression on the right to be correct for $n = N$, we adopt the notational convention $m_{N+1} + \cdots + m_N = 0$ which gives $M_N = M$ -- when all $N$ blocks have been removed, only the mass of the piston remains.  Just before slice $n$ is removed, the system is in equilibrium, and the internal energy of the gas is $U_{n-1} = cM_{n-1} gh_{n-1}$.  After slice $n$ is removed and the system has relaxed into the next equilibrium, the internal energy of the gas is $U_n = cM_n g h_n$, so the change in its internal energy is
\begin{align}\label{pecn}
	\Delta U_n = U_n - U_{n-1} &= cM_ng h_n - cM_{n-1}g  h_{n-1}.
\end{align}
During this process, the load starts and ends at rest, so the gas has performed work on the piston equal to the change in potential energy of the remaining mass $M_n$.  Since that mass undergoes a height change $h_n - h_{n-1}$, that work is
\begin{align}\label{wn}
	W_n = M_n g(h_n - h_{n-1}).
\end{align}
There is no heat transferred to the gas during this process, so the first law $\Delta U_n = Q_n-W_n$ reduces to $0 = \Delta U_n + W_n$.  Combining this with Eqs. \eqref{pecn} and \eqref{wn} gives
\begin{align}
	0 
	&= (1+c)M_n   gh_n - (cM_{n-1} + M_n)  gh_{n-1}.
\end{align}
Solving for $h_n$ and using the observation $M_{n-1}/M_n = 1 + m_n/M_n$ which follows from the definition \eqref{remain} gives a recursion relation for $h_n$;
\begin{align}\label{recur}
	h_n 
	&= \left(1+\frac{1}{\gamma}\frac{m_n}{M_n}\right)h_{n-1}, \qquad \gamma = \frac{c+1}{c}.
\end{align}
The quantity $\gamma>1$ is the adiabatic index of the gas.  For any real $m>0$ and integer $N > 0$, we define a length-$N$ partition of $m$ to be any sequence $\vec m = (m_1, \dots, m_N)$ of positive real numbers that sum to $m$.  Repeated application of the recursion relation \eqref{recur} gives an expression for the ratio $h_N/h_0$ in terms of a function $\h$ defined on the set of all partitions;
\begin{align}\label{hdef}
	\frac{h_N}{h_0} = \h(\vec m), \qquad \h(\vec m)
	\overset{\mathrm{def}}{=} \prod_{n=1}^N\left(1+\frac{1}{\gamma}\frac{m_n}{m_{n+1} + \cdots + m_N + M}\right).
\end{align}
As in the definition \eqref{remain} for $M_n$, we use the notational convention $m_{N+1} + \cdots + m_N = 0$. We have suppressed dependence of $\h$ on $\gamma$ and $M$ to keep the notation compact.  By virtue of Eq. \eqref{workfrac}, our goal of maximizing the stepwise work is equivalent to minimizing $\h$.  It turns out that understanding how $\h$ behaves as a function of the partition $\vec m$ benefits from an analysis of the reversible expansion starting and ending with the same loads $M+m$ and $M$ of the stepwise expansion.

\subsection{Height ratio for a reversible expansion}

A reversible, adiabatic expansion between an initial pressure-volume state $(P_0, V_0)$ and a final state $(P_\mathrm{rev},V_\mathrm{rev})$ can be mathematically represented by a smooth curve satisfying the differential form of the first law of thermodynamics for an adiabatic process; $dU = -PdV$.  Combining this with the equation of state \eqref{upv} yields the standard relation characterizing that curve; $P_0V_0^\gamma = P_\mathrm{rev}V_\mathrm{rev}^\gamma$.

For a reversible expansion to start and end at the same pressures as the stepwise expansion, it requires $P_0 = (m+M)g/A$ and $P_\mathrm{rev} = Mg/A$.  Moreover the initial and final heights satisfy $V_0 = Ah_0$ and $V_\mathrm{rev} = Ah_\mathrm{rev}$.  Combining these observations and a bit of algebra gives
\begin{align}\label{revr}
	\frac{h_\mathrm{rev}}{h_0} = \h_\mathrm{rev}(m), \qquad \h_\mathrm{rev}(m) \overset{\mathrm{def}}{=}\left(1+\frac{m}{M}\right)^{1/\gamma} .
\end{align}
As with the definition of $\h$, we omit explicitly indicating $\gamma$ and $M$ in the notation for $\h_\mathrm{rev}$ to keep the notation compact.  

\section{Numerical experiments}\label{explore}

To get a sense for how to minimize $\h$ and thus maximize the stepwise work, we do some numerical experiments. For all of our experiments, we take $c = 1/2$ (one spatial dimension), which gives $\gamma = 3$, but all results generalize.  The numerical experiments we perform would be instructive computational exercises for a thermodynamics course, but they could also serve as the basis for a computational physics project focusing on optimization of real-valued functions on high-dimensional spaces.

\subsection{Optimal ordering for fixed $N$}

We start with the following question: how does the order of the blocks in the stack impact $\h$, and thus the stepwise work?  We get a sense for the answer to this question by considering the extreme case $M = 0.01$, $m = 1$, and $\vec m = (0.99, 0.01)$. The load is 100 times the piston mass, and it's split into a pair of blocks, one of which is 99 times the mass of the other.  Definition \eqref{hdef} gives
\begin{align}
	\h(0.99, 0.01)
	&= \left(1 + \frac{1}{3}\frac{0.99}{0.01 + 0.01}\right)\left(1 + \frac{1}{3}\frac{0.01}{0.01}\right) 
	\approx (17.5)(1.3333).
\end{align}
When the more massive block is removed first, the piston height increases by a factor of 17.5.  Then the lighter block is removed, and it increases by a factor of about 1.33.  If the order of the blocks is switched, we get
\begin{align}
	\h(0.01, 0.99)
	&= \left(1 + \frac{1}{3}\frac{0.01}{0.99 + 0.01}\right)\left(1 + \frac{1}{3}\frac{0.99}{0.01}\right) 
	\approx (1.0033)(34).
\end{align}
The lighter block is removed first this time, and this causes the piston height to increase by a factor of 1.0033, then when the heavier block is removed, it increases by a factor of 34.  This is illustrated in Fig.  \ref{fig:small-big}.

\begin{figure}[htbp]
    \centering
    \begin{subfigure}[t]{0.48\textwidth}
    	\vspace{0pt}
        \centering
        \includegraphics[width=0.95\textwidth]{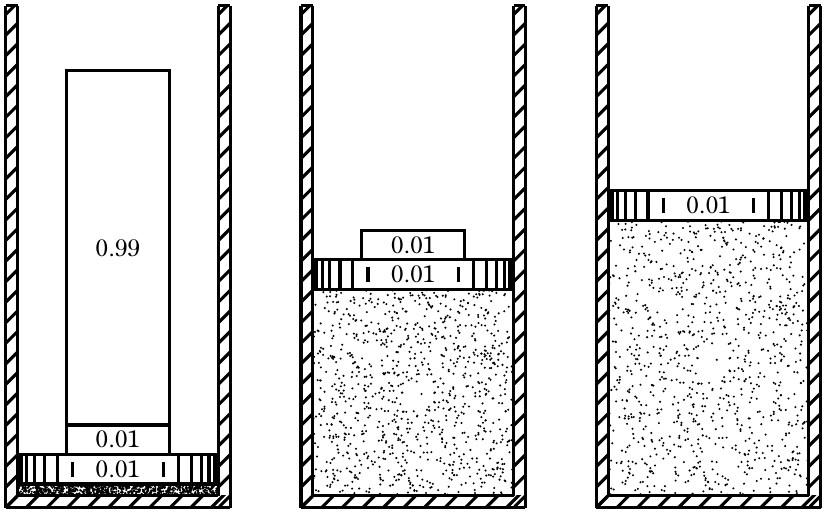}
        \label{fig:left}
    \end{subfigure}
    \hfill
    \begin{subfigure}[t]{0.48\textwidth}
    	\vspace{0pt}
        \centering
        \includegraphics[width=0.95\textwidth]{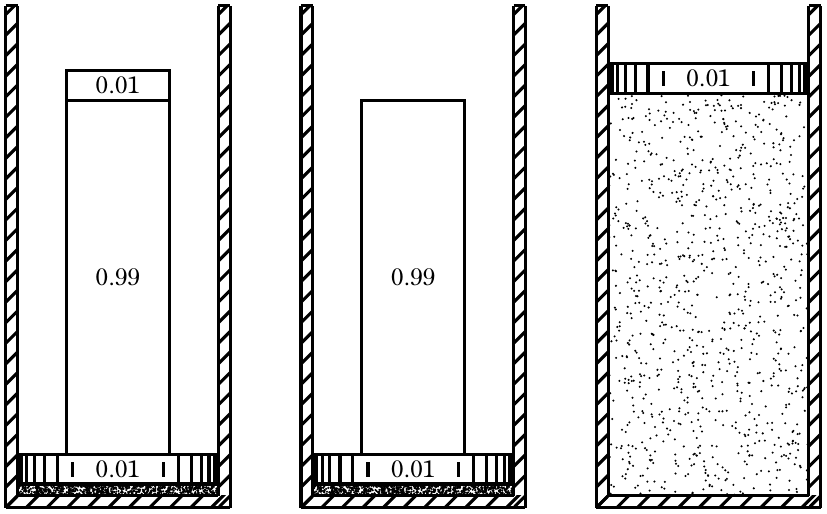}
        \label{fig:right}
    \end{subfigure}
  
    \caption{Stepwise expansions for $M = 0.01$ and the partitions $\vec m = (0.99, 0.01)$ and $\vec m = (0.01, 0.99)$.  Removing the heavier block first results in a smaller value of $\h$ and thus more stepwise work.}
    \label{fig:small-big}
\end{figure}

For both orders, removing the lighter block doesn't cause a substantial expansion because the removed load is not large compared to the remaining load.  But, when the heavier block is removed first, it leaves behind both the mass of the lighter block and the mass of the piston, while when it is removed second, it leaves only the mass of the piston.  This means that the ratio of the block's mass to the mass it leaves behind is about twice as big when it is removed second, and this causes it to increase the height of the piston by a factor of 34 instead of 17.5.  When a large mass is removed far down in the stack, the fractional change in mass is large upon its removal, and the piston height increases by a large factor.  Removing larger masses higher in the stack therefore tends to decrease $\h$ and thus increase the stepwise work.

Another small-$N$ numerical experiment confirms and extends this point.  We let $M=0.1$ and $m=1$, and we consider the length-3 partition $\vec m = (0.1, 0.3, 0.6)$.  We compute $\h$ for all $3! = 6$ permutations of the masses in $\vec m$.  Fig. \ref{perm} displays the results in order of decreasing $\h$.
\begin{figure}[h!]
	\centering\includegraphics[width = 0.95\textwidth]{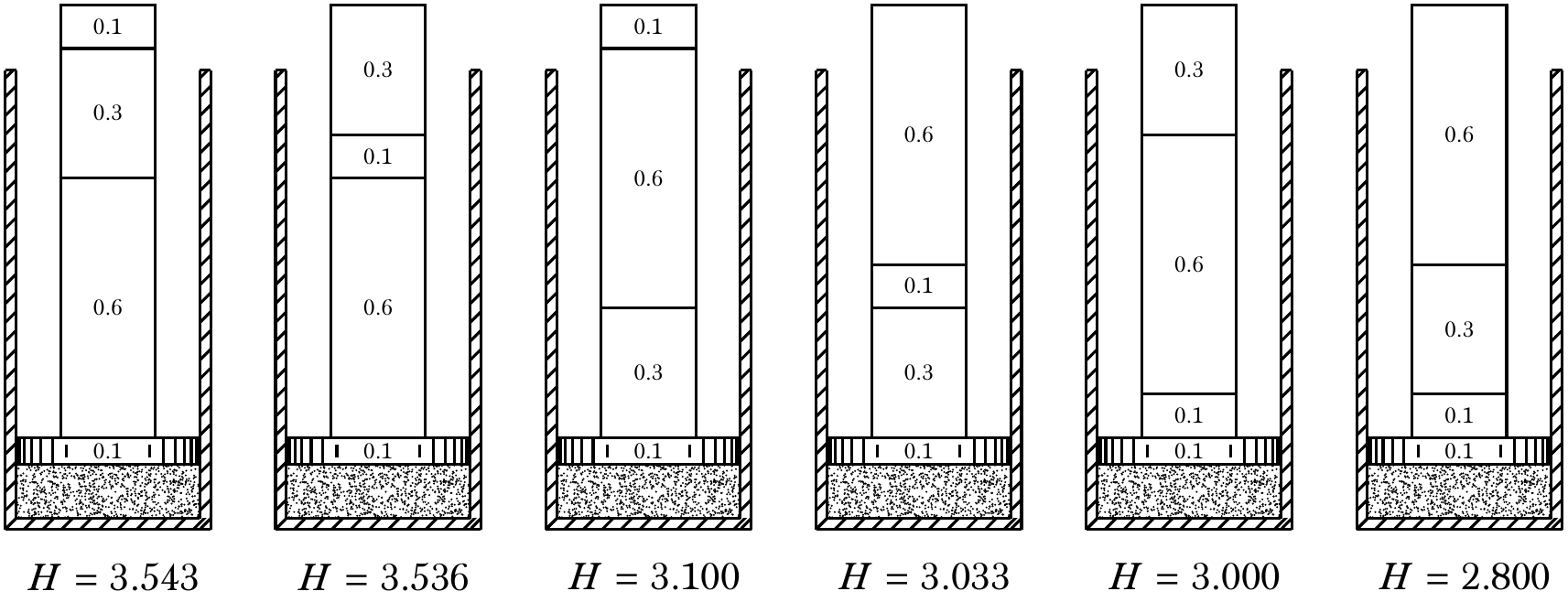}
	\caption{Permuting the blocks in the stack so that they are in non-increasing order of mass from top to bottom decreases $\h$ and therefore maximizes the stepwise work.}
	\label{perm}
\end{figure}
Inspection of these results reveals that the partition with decreasing mass from top to bottom has the smallest value of $\h$ and thus the largest stepwise work.  Moreover, for any two partitions that differ by swapping adjacent blocks, the one with the more massive block higher in the stack has the smaller value of $\h$.  Thermodynamically, these observations make sense.  A non-increasing order of mass helps control the irreversibility introduced by removing large masses later in the stack when this would cause a large fractional decrease in the load on the piston and thus an explosive expansion.  These experiments suggest the following properties which are true in general and are formally proven in Appendix \ref{proofs}.
\begin{property}[Swapping]
	Swapping a pair of adjacent blocks so that the larger mass block is on top leads to a smaller value of $\h$ and thus more stepwise work. (Theorem \ref{swapping})
\end{property}

\begin{property}[Optimal Order]
	Any partition with a non-increasing order of mass from top to bottom minimizes $\h$ and thus maximizes stepwise work. (Theorem \ref{order})
\end{property}

\subsection{Optimal partition for fixed $N$}

For a fixed number of blocks in a partition, re-ordering is not the only way to increase work.  One can also re-distribute mass among the blocks.  This leads to natural questions; if we put no constraints on the distribution or order, is there a partition that minimizes $\h$ and thus maximizes the stepwise work?  If so, is it unique? We again perform numerical experiments to get a sense for the answer.  We consider $M=0.1$, and $m=1$.  For each $N\in\{2,3, 4, 5, 6, 7\}$ we use a numerical optimization routine\footnote{The specific optimization algorithm used here is the sequential least squares programming (SLSQP) option for the \textsf{minimize} method in the SciPy Python library to optimize $\h$ subject to the constraints $m_1 > 0, \dots, m_N > 0$ and $m_1 + \cdots + m_N = m$ that characterize the set of partitions as a subset of $\reals^N$.} to search for a global minimum of $\h$.  The results are shown in Fig. \ref{fig:optimal}.
\begin{figure}[h!]
	\centering\includegraphics[width = 0.95\textwidth]{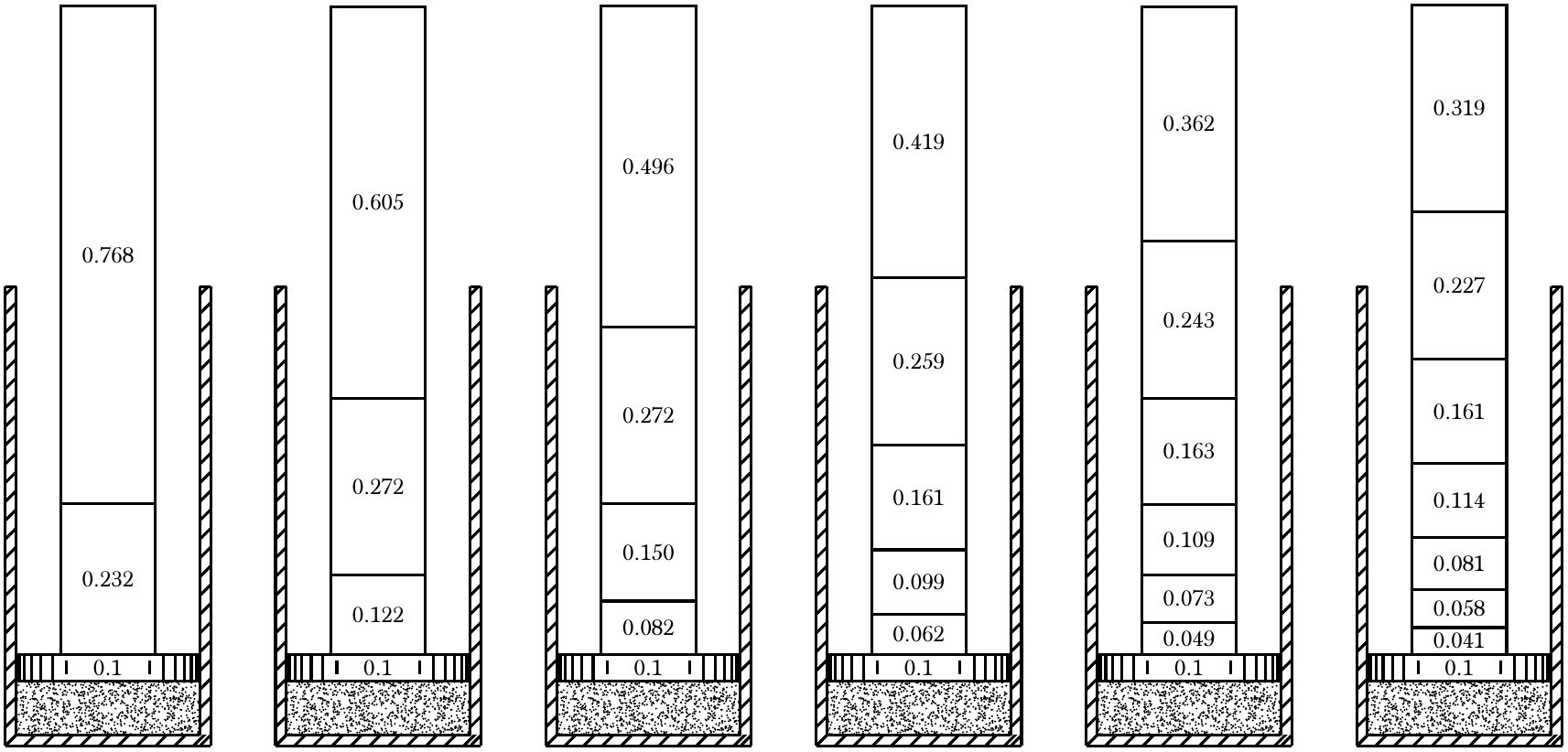}
	\caption{Numerically optimal partitions for $M=0.1$, $m=1$, and $N\in\{2,3, 4, 5, 6, 7\}$.}
	\label{fig:optimal}
\end{figure}
There is a unifying structure hidden in these numbers that was revealed by Andresen et al. who showed that a geometric progression of equilibrium pressures \eqref{geop} is a critical point of the stepwise work, and this implies a geometric progression of masses.  To see why, notice that after removing each block, the remaining mass $M_n$ determines the pressure exerted on the gas via $P_n = M_n g/A$, so $P_n = \alpha P_{n-1}$ implies $M_n = \alpha M_{n-1}$.  But a geometric progression of masses $M_n$ implies a geometric progression of the block masses $m_n$ because by definition \eqref{remain},
\begin{align}
	m_n = M_{n-1} - M_n = \alpha M_{n-2} - \alpha M_{n-1} = \alpha m_{n-1}.
\end{align}
If we return to our numerical experiment depicted in Fig. \ref{fig:optimal}, we can verify that adjacent masses in the stack differ by a constant ratio by computing the list of ratios $(m_1/m_2, \dots, m_{N-1}/m_N)$;
\begin{equation}
\begin{aligned}
	N = 2&: \quad(3.317) \\
	N = 3&: \quad(2.224, 2.224) \\
	N = 4&: \quad(1.821, 1.821, 1.821) \\
	N = 5&: \quad(1.615, 1.615, 1.615, 1.615) \\
	N = 6&: \quad(1.491, 1.491, 1.491, 1.491, 1.491) \\
	N = 7&: \quad(1.409, 1.409, 1.409, 1.409, 1.409, 1.409)
\end{aligned}
\end{equation}
These results reflect the following global optimality result that holds in general and was suggested in \cite{Andresen1977} but not proven:
\begin{property}[Optimal Partition]
	For a given $M$ and $m$, there is a unique global minimizer of $\h$ among all length-$N$ partitions, and thus there is a unique global maximizer of the stepwise work.  This partition consists of a geometric progression of masses corresponding to the pressures in \eqref{geop}. (Theorem \ref{abns})
\end{property}

\subsection{The reversible bound}\label{srz}

In the last section, we determined how one can maximize stepwise work for a given $N$, but now we relax the constraint of fixed $N$ and compare the stepwise work across all partitions of any length.  We start by building on our results in the last section. We compute the value of $\h$ for each of the fixed-$N$ optimal partitions in Fig. \ref{fig:optimal}, and we find the following:
\begin{equation}\label{ndec}
\begin{aligned}
	N = 2&: \quad\h(0.768, 0.232) \approx  3.141\\
	N = 3&: \quad\h(0.605, 0.272, 0.122) \approx  2.791\\
	N = 4&: \quad\h(0.496, 0.272, 0.150, 0.082) \approx  2.632\\
	N = 5&: \quad\h(0.419, 0.259, 0.161, 0.099, 0.062) \approx  2.542\\
	N = 6&: \quad\h(0.362, 0.243, 0.163, 0.109, 0.073, 0.049) \approx 2.484 \\
	N = 7&: \quad\h(0.319, 0.227, 0.161, 0.114, 0.081, 0.058, 0.041) \approx 2.444 
\end{aligned}
\end{equation}
As the number of blocks in the partition increases, the value of $\h$ decreases monotonically, and successive values of $\h$ get closer together.  Are they converging to a limiting value?  To numerically investigate convergence, we'd like to compute $\h$ for the fixed-$N$ optimal partition for a sequence of very large values of $N$.  For this purpose, it would be helpful to have a closed form expression for $\h$ for the optimal partition.  Fortunately, theorem \ref{abns} of Appendix \ref{proofs} gives an explicit expression for the masses in the fixed-$N$ optimal partition, and it results in a simple closed-form expression for $\h$.  Theorem \ref{abns} says the optimal partition for given $N$ is $\vec m^* = (m^*_1, \dots, m^*_N)$, where $\alpha$ satisfy $M_{n-1}^*/M_n^* = 1/\alpha$.  Definitions \eqref{hdef} and \eqref{remain} then give the following after a bit of algebra:
\begin{equation}
\begin{aligned}
	\h(\vec m^*)
	&= \left(1-\frac{1}{\gamma}+\frac{1}{\alpha\gamma}\right)^N.\label{opth}
\end{aligned}
\end{equation}
According to Theorem \ref{abns}, $\alpha = (M/(m+M))^{1/N}$, so if we take $M=0.1$ and $m=1$ to be consistent with the choices that led to \eqref{ndec}, we get
\begin{equation}\label{ndecbig}
\begin{aligned}
	N = 1&: \quad\h(\vec m^*) \approx 4.33333\\
	N = 10&: \quad\h(\vec m^*) \approx  2.37447\\
	N = 100&: \quad\h(\vec m^*) \approx  2.23827\\
	N = 1000&: \quad\h(\vec m^*) \approx  2.22540\\
	N = 10000&: \quad\h(\vec m^*) \approx  2.22412\\
	N = 100000&: \quad\h(\vec m^*) \approx 2.22399 \\
	N = 1000000&: \quad\h(\vec m^*) \approx 2.22398 \\
	N = 10000000&: \quad\h(\vec m^*) \approx 2.22398
\end{aligned}
\end{equation}
As expected from the trend in \eqref{ndec}, the sequence of $\h$ values seems to be converging to something as $N$ grows large, but to what?  When $N$ is large, the fixed-$N$ optimal partition consists of a large number of small-mass blocks.  For such a partition, the piston moves only a very small amount after each block is removed, so the system never strays too far from equilibrium.  We might therefore expect that $\h$ will converge on the value $\h_\mathrm{rev}$ -- the factor by which the height increases in a reversible expansion.  A quick computation using the definition \eqref{revr} confirms this; $\h_\mathrm{rev} = (1+m/M)^{1/\gamma} = (1 + 1/0.1)^{1/3} \approx 2.22398$.  To prove this in general, take the log of both sides of Eq. \eqref{opth}, define $\lambda = \ln(1+m/M)$, and do a bit of algebra with Taylor series to obtain
\begin{align}
	\ln\h(\vec m^*) = \frac{\ln\left(1 + \frac{1}{\gamma}(e^{\lambda/N} -1)\right)}{1/N} = \frac{\lambda}{\gamma} + O(N^{-1})
\end{align}
This implies that $\ln\h(\vec m^*)\to \lambda/\gamma$ as $N\to\infty$ and thus that $\h(\vec m^*)\to (1+m/M)^{1/\gamma}=H_\mathrm{rev}(m)$ as desired.

Is this limiting behavior just a special characteristic of the optimal partition, or is it generally true for any partition with a large number of small masses?  Let's look at a more generic sequence of partitions that becomes finer as $N$ increases.

\begin{figure}[h!]
	\centering\includegraphics[width = 0.95\textwidth]{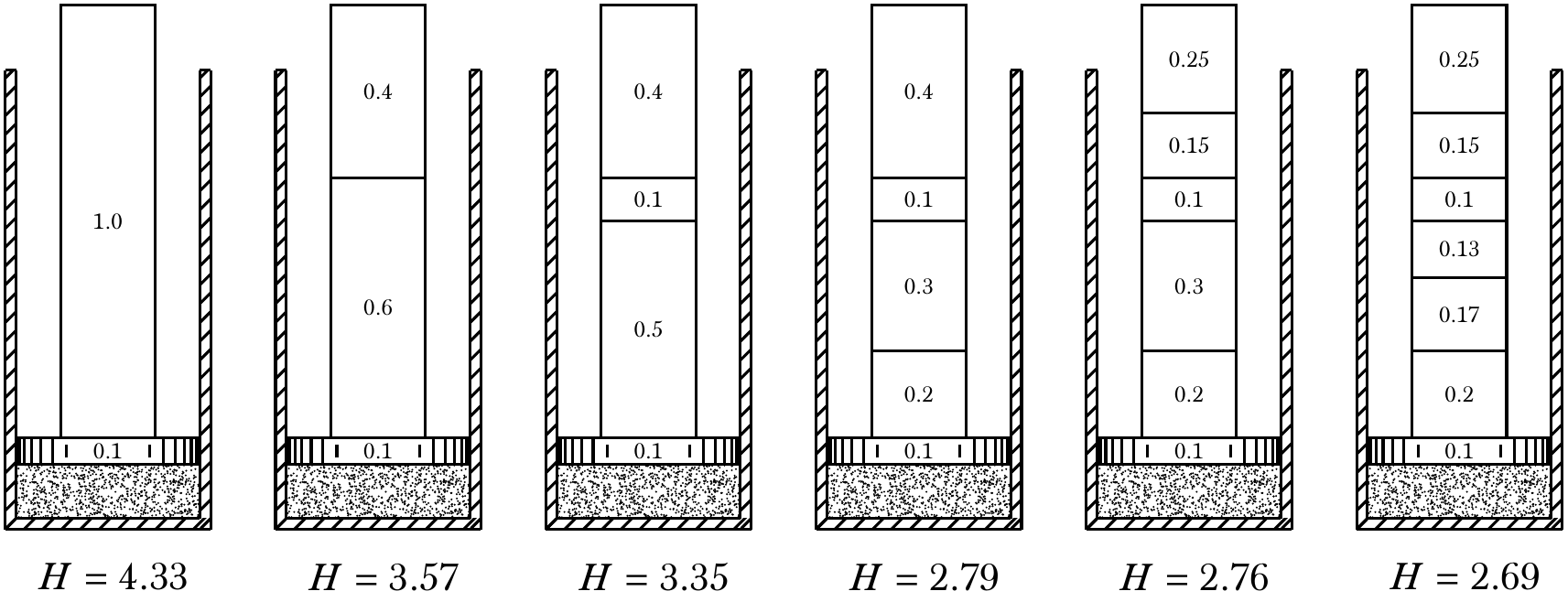}
	\caption{Iterative refinement of a partition with a fixed load $m = 1.0$.  As the partition is refined, the value of $\h$ decreases, indicating an increase in stepwise work.}
	\label{iteref}
\end{figure} 
 
We start with $N=1$ and compute $\h$.  We then iteratively refine the partition by selecting the largest-mass block and slicing it into two (not necessarily equal-mass) pieces.  Doing this by hand five times results in the sequence of partitions and corresponding $\h$ values of Fig. \ref{iteref}.  Each time a block is sliced in two, the value of $\h$ decreases, and the change in $\h$ also decreases. Does this process again converge to the reversible limit for large $N$? Iterating this process 1000 times by using a uniform random number to decide where to slice the largest block at each iteration results in the plot of $\h$ versus $N$ in Fig. 
\ref{scon}.  The plot confirms that slicing a mass in two reduces $\h$ and leads to the conjecture that if a partition has a fine mesh (has a small largest mass), it will have a small value of $\h$.  The plot also confirms the expected convergence to the reversible limit, and suggests that the reversible limit is a strict lower bound on $\h$.  These properties do indeed hold in general;

%\newpage

\begin{figure}[h!]
	\centering
	\begin{adjustbox}{max width=0.95\textwidth}
  		\input{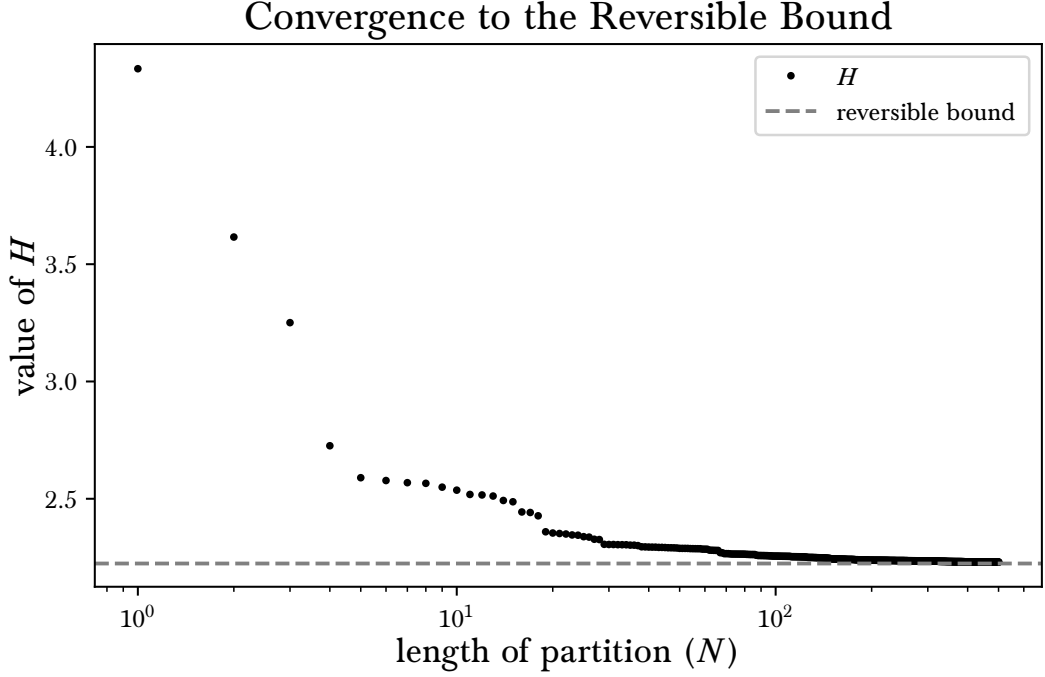}
	\end{adjustbox}
	\caption{Convergence to the reversible bound through random iterative refinement of a partition.}
	\label{scon}
\end{figure}

\begin{property}[Slicing]
	Refining a partition by slicing one of the masses in two decreases $\h$ and therefore increases the stepwise work. (Theorem \ref{slicing})
\end{property}

\begin{property}[Reversible Bound]
	$\h_\mathrm{rev}$ is a strict lower bound on $\h$ for all partitions $\vec m$, so the reversible adiabatic work strictly upper bounds stepwise adiabatic work. (Theorem \ref{revbound})
\end{property}

\begin{property}[Mesh Bound]
	$\h$ is bounded above by a quantity $\h_\mathrm{mesh}$ depending on the largest mass in the partition, so the stepwise work is bounded below by such a quantity. (Theorem \ref{meshbound})
\end{property}

\begin{property}[Zero Mesh Limit]
	The smaller the mesh (largest mass) of a partition, the closer $\h$ will be to $\h_\mathrm{rev}$, and thus the closer the stepwise work will be to the reversible limit. (Theorem \ref{zeromesh})
\end{property}

These properties all make physical sense.  When a block is sliced in two, the system won't go quite as far out of equilibrium after each mass is removed. The reversible bound property is consistent with the idea that irreversible processes are less efficient at extracting work than the corresponding reversible process, and the mesh bound is consistent with the idea that stepwise processes with smaller steps remain closer to equilibrium and are thus closer to reversible processes.  The reversible bound and mesh bound properties together imply the zero mesh limit property by the squeeze theorem.  

It's instructive to summarize all of these results in a single plot.
\begin{figure}[h!]
	\centering
	\begin{adjustbox}{max width=0.95\textwidth}
  		\input{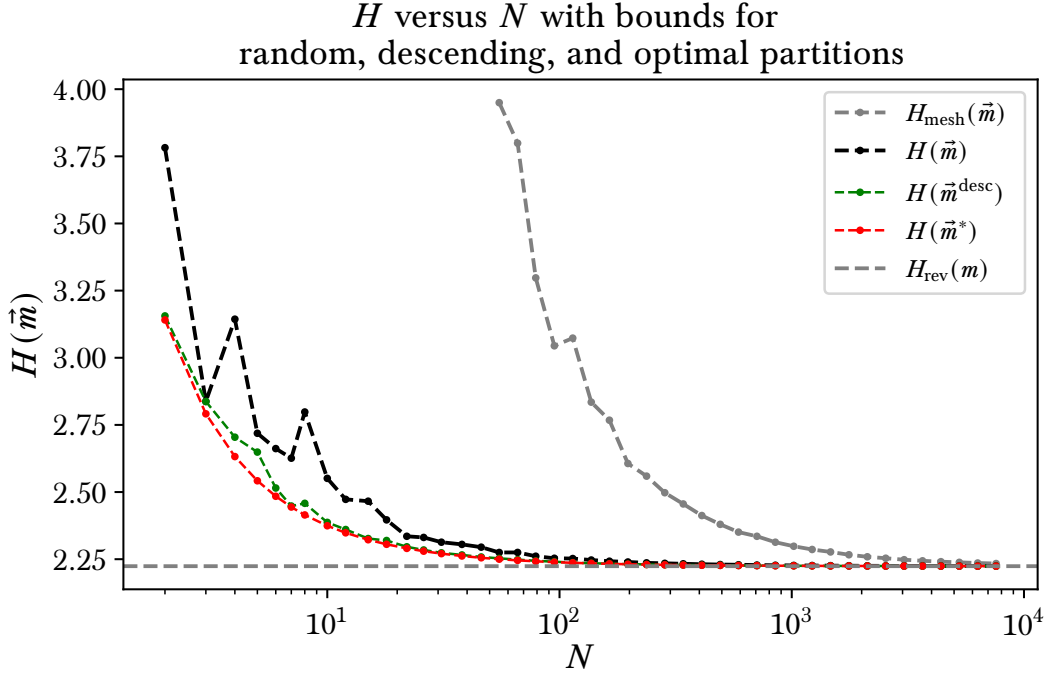}
	\end{adjustbox}
	\caption{For $\gamma = 3$, $M = 0.1$ and $m = 1$, the plot depicts the value of $\h(\vec m)$ versus $N$ for a number of different partitions $\vec m$.  Dashed lines are not data points but are included so trends are visually clear.}
	\label{rdo}
\end{figure}
For each $N$ in an approximately geometric progression from $N=2$ to $N=237$ (this gives approximately equally-spaced points on a log scale), we generate a random length-$N$ partition of $m$.  This gives a sequence of random partitions $\vec m_1, \vec m_2, \vec m_3, \dots, \vec m_N, \dots$.  For each $\vec m_N$ in the sequence, we compute the upper bound $\h_\mathrm{mesh}(\vec m_N)$ and $\h(\vec m_N)$.  For small $N$, the partition tends to contain large blocks, so $\h(\vec m_N)$ is large, and $\h_\mathrm{mesh}(\vec m_N)$ which is exponentially large (see Theorem \eqref{meshbound}) in the mesh (largest mass) $\mesh(\vec m)$.  We then permute the elements of $\vec m_N$ to obtain a descending-order partition $\vec m_N^\mathrm{desc}$ and compute $\h(\vec m_N^\mathrm{desc})$.  Since a descending order minimizes $\h$ for a given set of blocks, we expect this to decrease the value of $\h$.  Next, we compute the optimal partition $\vec m_N^*$ of Theorem \ref{abns} and $\h(\vec m_N^*)$.  This should yield an even smaller value of $\h$ since it's optimal among all finite-$N$ partitions.  Lastly, we compute the reversible bound $\h_\mathrm{rev}(m)$ which should be the lowest of all since it is a lower bound that can only be attained in the limit that the largest mass goes to zero.  We plot all of these five values of $\h$ versus $N$.  This results in Fig. \ref{rdo}.  At every $N$ the plot is consistent with the expected sequence of inequalities:
\begin{align}
	\h_\mathrm{mesh}(\vec m_N) > \h(\vec m_N) \geq \h(\vec m_N^\mathrm{desc}) \geq \h(\vec m_N^*) > \h_\mathrm{rev}(m).
\end{align}
As $N$ grows large, partitions whose mesh tends to zero converge to the reversible bound\footnote{If the largest mass in a partition remains large as $N\to\infty$, this convergence won't happen.  The randomly chosen partitions converge on the plot because for $N$ large, there is a low probability of picking a partition where any of the masses is large.} because in that limit, the upper bound $\h_\mathrm{mesh}$ converges to $\h_\mathrm{rev}$, so since all partitions have values of $\h$ between these two, their values get squeezed down to $\h_\mathrm{rev}$.

\section{Conclusion}

The purpose of this paper is to make the reversible adiabatic limit explicit in a simple finite-step model.  In this model, an ideal gas expands by lifting a piston as a load is removed in discrete chunks.  Each finite sequence of removals is irreversible, but the collection of all such finite-step processes gives a concrete way to ask how the reversible adiabatic value is approached.  The key simplification is that, for equilibrium states of the gas-piston-load system, the internal energy of the gas can be written as \(U=cM_{\mathrm{tot}}gh\).  This reduces the work calculation to the study of the final height ratio \(h_N/h_0=H(\vec m)\).  Since maximizing the work is equivalent to minimizing \(H\), the thermodynamic question becomes a finite-dimensional optimization problem over partitions of the load.

Numerical experiments reveal three main patterns, each of which is proved in the Appendix.  First, for a fixed collection of block masses, the gas does the most work when the blocks are removed in non-increasing order of mass.  Physically, this prevents a large block from being removed late in the process, when it would produce a large fractional drop in the remaining load and therefore a highly irreversible expansion.  Second, for fixed \(N\) and fixed total load, the unique globally optimal partition is the one for which the successive equilibrium pressures form a geometric progression.  This proves the global optimality of the pressure progression identified by Andresen et al. as a critical point.  Third, the reversible adiabatic expansion gives a strict upper bound on the work from every finite stepwise expansion, and this bound is approached whenever the mesh of the partition, the largest block mass, tends to zero.  This is consistent with the broader principle that refining irreversible step processes can reduce dissipation \cite{salamon2023}.

Figure~\ref{rdo} summarizes these results.  The random curve represents the value of $H$ from a typical finite partition with no attempt to optimize either the order or the block sizes.  The descending-order curve uses the same blocks but puts them in the best order, so it isolates the gain obtained purely from reordering.  The optimal curve represents the best possible \(N\)-block partition and shows the additional gain obtained by choosing the block sizes themselves.  The reversible curve is the limiting value that no finite stepwise process attains, but which can be approached by sufficiently fine partitions.  The mesh-bound curve is different in character: it is not a physical protocol, but an analytic upper bound on \(H\) that depends on the largest block.  It therefore emphasizes that the relevant condition for approaching reversibility is not merely that \(N\) be large, but that no individual step remain large.  When the mesh tends to zero, the mesh bound converges to \(H_{\mathrm{rev}}\), and the stepwise values are squeezed to the reversible limit.

This model also gives useful computational exercises for advanced students.  Starting from the recursion for \(H(\vec m)\), students can numerically discover that heavier blocks should be removed first, that the optimal finite-\(N\) partition is geometric, and that refinement of a partition drives the work toward the reversible value.  The subsequent proofs then show how these numerical observations follow from elementary but nontrivial mathematical ideas: sorting via adjacent swaps, convexity, Jensen's inequality, and the squeeze theorem.

As a natural extension of this work, one could compare the work-based optimization used here with entropy production optimization.  Moreover, the same questions could be asked for other equations of state.  These variants would preserve the central pedagogical theme: reversible thermodynamic processes need not remain abstract idealizations, they can be understood as limits of explicit irreversible processes whose finite-step structure is itself rich and instructive.

\appendix

\section{Proofs of general properties}\label{proofs}

For all proofs, we assume without loss of physical generality that $\gamma>1$, $m>0$, $M>0$, $N$ is a positive integer, and $\vec m$ is a partition of $m$.  

\begin{theorem}[Swapping]\label{swapping}
	Let $N$ be the length of $\vec m$, and choose $k\in\{1, \dots, N-1\}$. Let $\vec m^\mathrm{swap}$ be the re-ordering of $\vec m$ with $m_k$ and $m_{k+1}$ swapped;
	\begin{align}
		\vec m^\mathrm{swap}
		&= (m_1, \dots, m_{k-1}, m_{k+1}, m_k, m_{k+2}, \dots, m_N).
	\end{align}
	If $m_k < m_{k+1}$, then
	\begin{align}
		\h(\vec m) > \h(\vec m^\mathrm{swap}).
	\end{align}
\end{theorem}
\begin{proof}
	Inspecting definition \eqref{hdef}, we find that $\h(\vec m)$ and $\h(\vec m^\mathrm{swap})$ differ in only the factors with $n=k$ and $n=k+1$. If we take their ratio, all other factors cancel, and using the shorthand notation of \eqref{remain}, and the further shorthands $A = m_k$, $B = m_{k+1}$, and $C = M_{k+1}$, we obtain
	\begin{align}
		\frac{\h(\vec m)}{\h(\vec m^\mathrm{swap})}
		&= \frac{\left(1 + \frac{1}{\gamma}\frac{A}{B + C}\right)\left(1 + \frac{1}{\gamma}\frac{B}{C}\right)}{\left(1 + \frac{1}{\gamma}\frac{B}{A + C}\right)\left(1 + \frac{1}{\gamma}\frac{A}{C}\right)}
		=\frac{g(A,B)}{g(B,A)}
	\end{align}
	where the last equality defines the function $g$.  After some tedious algebra, the difference between the numerator and denominator is found to be
	\begin{align}
		g(A,B) - g(B,A)
		&= \frac{AB\,(\gamma-1)}{\gamma^2\,C\,(A+C)\,(B+C)}\cdot (B-A).
	\end{align}
	Since $A$, $B$, and $C$ are positive, $\gamma > 1$, and $A < B$, the expression on the right is positive, so $g(A,B) > g(B, A)$ which gives the desired inequality.
\end{proof}

If $\vec m$ is a length-$N$ partition and $\sigma$ is a permutation of the set $\{1, \dots, N\}$, then we define $\vec m_\sigma = (m_{\sigma(1)}, \dots, m_{\sigma(N)})$.

\begin{theorem}[Optimal Order]\label{order}
    Let $N$ be the length of $\vec m$.  Let $\sigma$ be any permutation of $\{1, \dots, N\}$ for which
	\begin{align}
		m_{\sigma(1)} \geq m_{\sigma(2)} \geq\cdots \geq m_{\sigma(N)}
	\end{align}
	then
	\begin{align}
		\h(\vec m) \geq \h(\vec m_{\sigma}).
	\end{align}
\end{theorem}

\begin{proof}
Applying the standard bubble sort algorithm to $\vec m$ yields a permutation $\tau=\sigma_1\circ\cdots\circ\sigma_j$
such that $\vec m_\tau$ is non-increasing and each $\sigma_i$ is an adjacent transposition
that swaps an adjacent pair
$m_k<m_{k+1}$ (so $\h$ decreases by Theorem~\ref{swapping}). Since $\vec m_\tau=\vec m_\sigma$
as vectors (they differ only by permuting equal masses), $\h(\vec m_\tau)=\h(\vec m_\sigma)$.
Repeatedly applying Theorem~\ref{swapping} therefore gives
\[
\h(\vec m)\ge \h(\vec m_{\sigma_1})\ge \h(\vec m_{\sigma_1\circ\sigma_2})\ge \cdots \ge
\h(\vec m_{\sigma_1\circ\cdots\circ\sigma_j})
= \h(\vec m_\tau)=\h(\vec m_\sigma),
\]
as desired.
\end{proof}

 For the results that follow, we use the following notation for the set of all length-$N$ partitions of total mass $m$  
\begin{align}\label{parset}
		 \partition_{m, N} = \{\vec m\in\reals^N \mid m_1 > 0, \dots, m_N > 0, \quad m_1 + \cdots + m_N = m\}.
	\end{align}

\begin{theorem}[Optimal Partition]\label{abns}
	Let $\vec m^*$ be the length-$N$ partition defined as follows:
	\begin{align}\label{geostack}
	 m^*_n = \alpha  m^*_{n-1}, \qquad  m^*_1 = (1-\alpha)(m+M), \qquad \alpha = \left(\frac{M}{m+M}\right)^{1/N}
	\end{align}
	If $\vec m\in\partition_{m,N}$ and $\vec m\neq\vec m^*$ then
	\begin{align}
		\h(\vec m) > \h(\vec m^*).
	\end{align}
\end{theorem}

\begin{proof}
	Let $\scl(\vec m) = \ln\h(\vec m)$.  Since the exponential function is strictly increasing, it suffices to show that for all $\vec m\in\partition_{m,N}$ with $\vec m\neq \vec m^*$,
	\begin{align}
		\scl(\vec m) > \scl(\vec m^*).
	\end{align}
	We do this by composing $\scl$ with an invertible mapping that converts it into a strictly convex function.  Let $f$ be the function that maps each $\vec m\in\partition_{m,N}$ to $\vec x\in\reals^N$ as follows: 
	\begin{align}\label{xdef}
		x_n = \ln\left(\frac{m_n + \cdots + m_N + M}{m_{n+1} + \cdots + m_N + M}\right).
	\end{align}
	Since each $m_n > 0$, each $x_n > 0$.  Moreover, using the notation of Eq. \eqref{remain}, we get $\alpha = (M_N/M_0)^{1/N}$ which implies that $\alpha^N = M_N/M_0$, so since $x_n = \ln(M_{n-1}/M_n)$, we get 
	\begin{align}\label{xsum}
		\sum_{n=1}^N x_n
		= \ln\left(\frac{M_0}{M_1}\frac{M_1}{M_2}\cdots \frac{M_{N-1}}{M_N}\right)
		=\ln\left(\frac{M_0}{M_N}\right)
		= -N\ln \alpha.
	\end{align}
	These observations imply that $f$ maps $\partition_{m,N}$ into the following  set.
	\begin{align}
		\mathscr X_{m,N} 
		= \left\{\vec x\in\reals^N\mid x_1>0,\cdots, x_N>0, \quad x_1 + \cdots + x_N = -N\ln\alpha\right\}.
	\end{align}
	We now show that $f:\partition_{m,N} \to\mathscr X_{m,N}$ is a bijection.  Its definition \eqref{xdef} implies the recursion
	\begin{align}
		M_{n-1} 
		&= e^{x_n} M_n 
		= e^{x_n}e^{x_{n+1}}M_{n+1}
		\cdots
		 = e^{x_n}e^{x_{n+1}}\cdots e^{x_N} M_N
		 = e^{x_n + \cdots + x_N} M.
	\end{align}
	and thus
	\begin{align}\label{mfx}
		m_n 
		= M_{n-1} - M_n
		= (e^{x_n + \cdots + x_N} - e^{x_{n+1} + \cdots + x_N}) M.
	\end{align}
	For any $\vec x\in \mathscr X_{m,N}$ this equation gives a $\vec m$ with $m_n > 0$ and $m_1 + \cdots + m_N = M_0 - M_N = m$, so it defines a mapping $g:\mathscr X_{m,N}\to\partition_{m,N}$.  Some algebra confirms that $g\circ f$ is the identity mapping on $\partition_{m,N}$ while $f\circ g$ is the identity mapping on $\mathscr X_{m,N}$, so $f$ is a bijection of $\partition_{m, N}$ onto $\mathscr X_{m,N}$ with $f^{-1} = g$.  Now define $\bar\scl:\mathscr X_{m,N}\to\reals$ by $\bar\scl = \scl\circ f^{-1}$, then unraveling definitions gives 
	\begin{align}
		\bar\scl(\vec x) = \sum_{n=1}^N \phi(x_n), \qquad \phi(x) = \ln\left(1-\frac{1}{\gamma} + \frac{e^{x}}{\gamma}\right).
	\end{align}
	Explicit computation shows that for all $x\in\reals$,
	\begin{align}
		\phi''(x)
		= \frac{(\gamma - 1)e^{x}}{(\gamma-1 + e^{x})^2} > 0,
	\end{align}
	so $\phi$ is strictly convex on the real line.  Jensen's inequality then implies that for any $\vec x\in\mathscr X_{m,N}$,
	\begin{align}
		\frac{1}{N}\bar\scl(\vec x) = \frac{1}{N}\sum_{n=1}^N \phi(x_n) \geq \phi\left(\sum_{n=1}^N \frac{x_n}{N}\right),
	\end{align}
	with equality if and only if $x_1 = x_2 = \cdots = x_N$.  In light of \eqref{xsum}, the condition that the components of $\vec x$ are all equal is equivalent to $Nx_n = -N\ln\alpha$.  Thus if we define $x_n^* = -\ln\alpha$, then for any $\vec x\in\mathscr X_{m,N}$ with $\vec x \neq \vec x^*$, we have
	\begin{align}
		\bar\scl(\vec x) > N\phi\left(\sum_{n=1}^N \frac{x_n}{N}\right) =\bar\scl(\vec x^*).
	\end{align}
	If we define $\vec m^* = f^{-1}(\vec x^*)$, we obtain
	\begin{align}\label{elli}
		\scl(f^{-1}(\vec x)) = \bar\scl(\vec x) > \bar\scl(\vec x^*) = \scl(f^{-1}(\vec x^*)) = \scl(\vec m^*).
	\end{align}
	Now consider $\vec m\in\partition_{m,N}$ with $\vec m\neq \vec m^*$, then bijectivity of $f$ implies that there is an $ \vec x \neq \vec  x^*$ with $\vec x = f(\vec m)$, and thus \eqref{elli} gives
	\begin{align}
		\scl(\vec m) > \scl(\vec m^*).
	\end{align}
	So $\vec m^* = f^{-1}(\vec x^*)$ is the unique global minimizer of $\scl$ on $\partition_{m,N}$.  But letting $M_n^* = m^*_{n+1} + \cdots + m^*_N + M$, we have $\ln(M_{n-1}^*/M_n^*) = x_n^* =  -\ln\alpha$, so $M_{n-1}^*/M_n^* = 1/\alpha$, and thus $M_n^* = \alpha^n M_0 = \alpha^n(m+M)$ and
	\begin{align}
		m_n^* 
		= M_{n-1}^* - M_n^*
		= (\alpha^{n-1} - \alpha^n)(m+M)
		= \alpha^{n-1}(1-\alpha)(m+M).
	\end{align}
	This implies the desired recursion \eqref{geostack}.	 
	 
\end{proof}

The partition $\vec m^*$ of the last theorem is precisely the mass partition whose associated equilibrium pressures satisfy \eqref{geop} of \cite{Andresen1977}.  To confirm this, we compute the sequence of equilibrium pressures of the stepwise expansion corresponding to $\vec m^*$ and verify that they satisfy \eqref{geop}.  The pressures $P_0, P_1, \dots, P_N$ are   $P_n = M_n g/A$, so 
\begin{align}\label{alpha}
	\alpha 
	= \left(\frac{M}{m+M}\right)^{1/N}
	= \left(\frac{M_N}{M_0}\right)^{1/N}
	= \left(\frac{M_Ng/A}{M_0g/A}\right)^{1/N}
	= \left(\frac{P_N}{P_0}\right)^{1/N},
\end{align}
and since $M_n^*  = \alpha M_{n-1}^*$ for the partition  $\vec m^*$, we get $P_n = \alpha P_{n-1}$.  These are precisely the conditions in \eqref{geop}.

\begin{theorem}[Slicing]\label{slicing}
	Let $N$ be the length of $\vec m$, choose $k\in\{1, \dots, N\}$, and let $A$ and $B$ be positive real numbers satisfying $A + B = m_k$.  If 
	\begin{align}
		\vec m^\mathrm{sliced} = (m_1, \dots, m_{k-1}, A, B, m_{k+1}, \dots, m_N)
	\end{align}
	then
	\begin{align}
		\h(\vec m) > \h(\vec m^\mathrm{sliced}).
	\end{align}

\end{theorem}
\begin{proof}
	The only difference between $\h(\vec m)$ and $\h(\vec m^\mathrm{sliced})$ is that the $k^\mathrm{th}$ factor in $\h(\vec m)$ gets split into a product of factors involving $A$ and $B$.  If we take the ratio of $\h(\vec m)$ and $\h(\vec m^\mathrm{sliced})$, all other common factors cancel, so using the shorthand $C = M_k$, we get  
	\begin{align}
		\frac{\h(\vec m)}{\h(\vec m^\mathrm{sliced})}
		&= \frac{\left(1+\frac{1}{\gamma}\frac{A + B}{C}\right)}{\left(1+\frac{1}{\gamma}\frac{A}{B + C}\right)\left(1+\frac{1}{\gamma}\frac{B}{C}\right)} 
	\end{align}
	Some tedious but straightforward algebra then gives
	\begin{align}
		\h(\vec m)
		&= \left(1 + \frac{AB(\gamma-1)}{\left(\gamma(B+C)+A\right)\left(\gamma C+B\right)}\right) \h(\vec m^\mathrm{sliced})
	\end{align}
	Since $\gamma > 1$, and $A$, $B$, and $C$ are positive, the expression in the parentheses on the right is larger than 1, and this gives the desired result.
\end{proof}

For the next theorem, it's useful to explicitly indicate the mass $M$ of the piston in the notations for $\h$ and $\h_\mathrm{rev}$, so instead of $\h(\vec m)$ and $\h_\mathrm{rev}(m)$, we write $\h(\vec m, M)$ and $\h_\mathrm{rev}(m, M)$.

\begin{theorem}[Reversible Bound]\label{revbound}
	\begin{align}
	\h(\vec m, M) > \h_\mathrm{rev}(m, M)
\end{align}
\end{theorem}
\begin{proof}
Let $N$ be the length of $\vec m$.  We proceed by induction on $N$. The base case $N=1$ is
\begin{align}\label{base}
	1+\frac{1}{\gamma} \frac{m}{M} > \left(1+\frac{m}{M}\right)^{1/\gamma}
\end{align}
Since $0 < 1/\gamma < 1$ and $m/M>0$, this inequality follows from the inequality $1+\beta x > (1+x)^\beta$ for $x>0$ and $0<\beta<1$.  This inequality follows from the fact that $(1+x)^\beta$ is strictly concave with $1 + \beta x$ its tangent line at $x=0$.  Now assume the bound is true for $N=k$, then for $N = k+1$ we have
\begin{equation}
\begin{aligned}
	\h&((m_1, \dots, m_{k+1}), M) \\
	&= \prod_{n=1}^{k+1}\left(1+\frac{1}{\gamma}\frac{m_n}{m_{n+1} + \cdots + m_{k+1} + M}\right) \\
	&= \prod_{n=1}^k\left(1+\frac{1}{\gamma}\frac{m_n}{m_{n+1} + \cdots + m_k + (m_{k+1} + M)}\right)\left(1+\frac{1}{\gamma}\frac{m_{k+1}}{M}\right) \\
	&= \h((m_1, \dots, m_k), m_{k+1} + M)\h((m_{k+1}), M)
\end{aligned}
\end{equation}
If we apply the base case $N=1$ to the factor $\h((m_{k+1}), M)$ and the inductive hypothesis about the case $N=k$ to the factor $\h((m_1, \dots, m_k), m_{k+1} + M)$, we get
\begin{equation}
\begin{aligned}
	\h((m_1, \dots, m_{k+1}), M) 
	&> \h_\mathrm{rev}(m_1 + \cdots + m_k, m_{k+1} + M)\h_\mathrm{rev}(m_{k+1}, M)\\
	&= \left(1+\frac{m_1 + \cdots +m_k}{m_{k+1} + M}\right)^{1/\gamma}\left(1+\frac{m_{k+1}}{M}\right)^{1/\gamma} \label{bprod}\\
	&= \left(\frac{m_1 + \cdots + m_{k+1} + M}{m_{k+1} + M}\cdot\frac{m_{k+1} + M}{M}\right)^{1/\gamma} \\
	&= \left(1+\frac{m_1 + \cdots + m_{k+1}}{M} \right)^{1/\gamma} \\
	&= \h_\mathrm{rev}(m_1 + \cdots + m_{k+1}, M)
\end{aligned}
\end{equation}
which is the desired inequality.
\end{proof}
For our last two theorems, we define the mesh of a partition $\vec m = (m_1, \dots, m_N)$ as the mass of the largest slice;
\begin{align}
	\mesh(\vec m) = \max\{m_1, \dots, m_N\}.
\end{align}

\begin{theorem}[Mesh Bound]\label{meshbound}
	If we define
	\begin{align}
		\h_\mathrm{mesh}(\vec m) = \exp\left(\frac{1}{2\gamma}\frac{m}{M} \frac{\mesh(\vec m)}{M}\right)\cdot \h_\mathrm{rev}(m)
	\end{align}
	then
	\begin{align}
		\h(\vec m) < \h_\mathrm{mesh}(\vec m).
	\end{align}
\end{theorem}

\begin{proof}
	The idea behind this proof is to take the logarithm of $\h(\vec m)$, prove an inequality for the resulting expression, and then exponentiate both sides.  If we define $\rho_n = m_n/M_n$ then
	\begin{align}\label{lrem}
		\ln \h(\vec m) 
		&= \sum_{n=1}^{N} \ln\left(1 + \frac{\rho_n}{\gamma}\right) 
		= \frac{1}{\gamma}\sum_{n=1}^N\ln\left(1+\rho_n\right) + R_N,
	\end{align}
	where we have defined
	\begin{align}\label{remainder}
		R_N = \sum_{n=1}^N\ln\left(1+\frac{\rho_n}{\gamma}\right) - \frac{1}{\gamma}\sum_{n=1}^N\ln\left(1+\rho_n\right).
	\end{align}
	But $1+\rho_n = M_{n-1}/M_n$ by definition \eqref{remain}, so since $M_0 = m+M$ and $M_N = M$, we get
	\begin{equation}
	\begin{aligned}
		\sum_{n=1}^{N} \ln(1+\rho_n)
		%&= \ln\left(\prod_{n=1}^{N} (1+\rho_n)\right)
		&= \ln\left(\frac{M_0} {M_1}\frac{M_1}{M_2}\cdots \frac{M_{N-1}}{M_N}\right) 
		%&= \ln\left(\frac{M_0}{M_N}\right)
		%= \ln\left(\frac{m+M}{M}\right) 
		= \ln\left(1+\frac{m}{M}\right). \label{lsum}
	\end{aligned}
	\end{equation}
	Combining \eqref{lrem} and \eqref{lsum} gives
	\begin{equation}
	\begin{aligned}
		\ln \h(\vec m) 
		&= \frac{1}{\gamma}\ln\left(1+\frac{m}{M}\right) + R_N.
	\end{aligned}
	\end{equation}
	We now bound $R_N$ from above which also allows us to bound $\h(\vec m)$ from above. Since $\rho_n$ and $\gamma$ are positive, the standard inequalities $x - x^2/2 < \ln(1+x) < x$ for $x>0$ imply that
	\begin{align}
		\ln\left(1 + \frac{\rho_n}{\gamma}\right) < \frac{\rho_n}{\gamma}, \qquad -\frac{1}{\gamma}\ln(1+\rho_n) < -\frac{1}{\gamma}\left(\rho_n - \frac{(\rho_n)^2}{2}\right).
	\end{align}
	Using these inequalities in \eqref{remainder}, and noting that $\rho_n =m_n/M_n\leq m_n/M$ and $m_n \leq \mesh(\vec m)$, we get
	\begin{equation}
	\begin{aligned}
		R_N
		&< \frac{1}{2\gamma}\sum_{n=1}^{N} (\rho_n)^2 
		\leq \frac{1}{2\gamma }\sum_{n=1}^{N} \left(\frac{m_n}{M}\right)^2 
		\leq \frac{1}{2\gamma M^2}\sum_{n=1}^N \left(\mesh(\vec m) \cdot m_n\right) \\
		&= \frac{1}{2\gamma M^2}\cdot \mesh(\vec m)\cdot \sum_{n=1}^N m_n 
		= \frac{m}{2\gamma M^2}\mesh(\vec m).
	\end{aligned}
	\end{equation}
	Putting this all together gives
	\begin{align}
		\ln \h(\vec m) < \frac{1}{\gamma}\ln\left(1+\frac{m}{M}\right) + \frac{m}{2\gamma M^2}\mesh(\vec m).
	\end{align}
	Finally, we exponentiate both sides of this inequality and invoke strict monotonicity of the exponential and definition \eqref{revr} to obtain the desired inequality.
\end{proof}

\begin{theorem}[Zero Mesh Limit]\label{zeromesh}
  Let $\vec m_1, \vec m_2, \dots $ be a sequence of partitions of $m$ for which $\mesh(\vec m_j)\to 0$ as $j\to\infty$, then
	\begin{align}
		\lim_{j\to\infty} \h(\vec m_j) = \h_\mathrm{rev}(m).
	\end{align}
\end{theorem}
\begin{proof}
	Theorem \ref{revbound} bounds $\h(\vec m_j)$ from below, and Theorem \ref{meshbound} bounds it from above.  Putting these together gives
	\begin{align}
		\h_\mathrm{rev}(m) 
		< 
		\h(\vec m_j) 
		<
		\exp\left(\frac{1}{2\gamma}\frac{m}{M} \frac{\mesh(\vec m_j)}{M}\right)\cdot \h_\mathrm{rev}(m).
	\end{align}
	for every $\vec m_j$ in the sequence.  As $j\to\infty$, the mesh of $\vec m_j$ goes to zero by assumption, so since $\gamma$, $m$, and $M$ are fixed, the exponential factor approaches $1$, so the right-hand side approaches $\h_\mathrm{rev}(m)$.  The squeeze theorem then gives the desired result.
\end{proof}

\section*{Author declarations} The author has no conflicts to disclose.

\bibliographystyle{unsrt}
\bibliography{refs}

\end{document}